\documentclass[10pt]{elsarticle}

\usepackage{amsthm,amsmath,amsfonts,amssymb,amscd,mathrsfs}
\usepackage{txfonts}
\usepackage{supertabular,soul}
\usepackage[usenames,dvipsnames]{xcolor}
\usepackage{tikz, graphicx,color,geometry}
 \usepackage{multirow}
 \usetikzlibrary{arrows}

\usepackage{bbm} 
\usepackage{hyperref}
\usepackage{yfonts}
\usepackage{eucal}
\usepackage{overpic}
\usetikzlibrary{calc}
\usepackage{enumitem}

\newcommand{\comment}[1]{}

\makeatletter
\def\ps@pprintTitle{%
\let\@oddhead\@empty
\let\@evenhead\@empty
\let\@oddfoot\@empty
\let\@evenfoot\@oddfoot}
\makeatother

\newtheorem{theorem}{Theorem}
\newtheorem{definition}{Definition}
\newtheorem{thm}{Theorem}[section]

\newtheorem{prop}[thm]{Proposition}

\newtheorem{example}[thm]{Example}

\theoremstyle{remark}

\providecommand*{\propertyautorefname}{Property}

\let\oldmarginpar\marginpar
\renewcommand\marginpar[1]{\oldmarginpar[\raggedleft\footnotesize #1]%
{\raggedright\footnotesize #1}}

\begin{document}
\begin{frontmatter}

\date{\today}

\title{Multiparticle Dynamics on the Triangular Lattice in Interacting Media}

\author[stewart]{Stewart McGinnis}
\address[stewart]{Department of Mathematics, Brigham Young University, Provo, UT 84602, USA, telestew@yahoo.com }
\author[ben]{Benjamin Webb}
\address[ben]{Department of Mathematics, Brigham Young University, Provo, UT 84602, USA, bwebb@mathematics.byu.edu}

\begin{abstract}
We study the motion of $N$ particles moving on a two-dimensional triangular lattice, whose sites are occupied by either left or right rotators. These rotators deterministically scatter the particles to the left (right), changing orientation from left to right (right to left) after scattering a particle. This interplay between the scatterers and the particle's motion causes a single particle to propagate in one direction away from its initial position. For multiple particles we show that the particles' dynamics can be vastly different. Specifically, we show that a particle can become entangled with another particle potentially causing the particle's trajectory to become periodic and that this can happen when the particles have the same or differing speeds. We describe two classes of periodic orbits based on the particles initial velocities. We also describe how a particle with an unbounded past trajectory implies that some, possibly other, particle(s) has an unbounded future trajectory in this and other related multiparticle models.
\end{abstract}


\end{frontmatter}

\section{Introduction}\label{sec:1}

This paper continues the investigation of a particular Lorentz lattice gas (LLG) system considered in \cite{Grosfils99}. Before describing this model in detail, we note that in a standard LLG, a single particle moves along the bonds of a lattice, from lattice site to lattice site. When the particle arrives at a lattice site, it encounters a scatterer that modifies the particle's motion according to a given scattering rule.

The reason we study such systems is to understand the basic principles that underly dynamic processes such as diffusion, propagation, etc. \cite{Ruijgrok88,Wang94,Wang95.1,Wang95.3,Kong90.2,Kong90.1,Meng94}. For simplicity, the study of a particle's motion on a lattice is a natural choice, since a lattice has both a discrete structure and a high degree of regularity. In such systems the type of scattering rules that have been investigated are physically motivated rules such as rotators, mirrors, etc. \cite{Grosfils99,Wang95.1,Bunimovich93,Cao97,Wang95.2,Kong89}, which are used to model a particle moving through various types of media.

As mentioned, in a LLG when the particle arrives at a lattice site it encounters a scatterer that modifies its motion according to a given scattering rule. Depending on the scattering rule, each scatterer can also have one of a number of orientations. Moreover, the orientation of each scatterer can be fixed or may change depending on the given scattering rule. The initial orientation of each scatterer is called the LLG's initial configuration of scatterers, which models the medium through which the particle moves. The trajectory of a particle is then determined by the particular choice of (i) lattice, (ii) scattering rule, and (iii) initial configuration of scatters on the lattice. In previous studies, a wide variety of dynamics has been observed in such systems, depending on the choice of these three features (see, for instance, \cite{Wang95.1,Wang95.3,Kong90.2,Wang95.2,Bunimovich91}).

In the case that there is a scatterer at each lattice site and each scatterer is \emph{fixed}, i.e. is not affected by the particle's motion, the problem of determining the particle's motion through the lattice is related to problems in percolation theory \cite{Cao97,Bunimovich91,Ziff91,Herrmann83}. When the scatterers are \emph{not fixed}, as is the case in this paper, and are affected by the particle, the particle's motion is a much more dynamic process and has connections to problems in kinetic theory \cite{Bunimovich2002,Turcotte1999,Velzen1991}.

The reason a single particle is typically studied in a LLG is its relevance in certain systems. Originally, H. A. Lorentz assumed in modeling conductance that electrons passing through a conductor move independently of each other \cite{Lorentz05}. Under this assumption it is sufficient to study the dynamics of a single electron. However, if the single particle's motion effects the orientation of the scatterers in an LLG then, even if the particles do not directly interact, they may indirectly interact through their influence on the system's scatterers.

In a number of previously considered LLGs the system's single particle does interact the system's medium, i.e. orientation of scatterers, as the particle moves through the lattice \cite{Grosfils99,Bunimovich93,Chisholm14,Baker13}. In these models scatterers are what are referred to as \emph{flipping scatterers} that change orientation as the particle collides with them. These are some of the simplest models in which there is an interaction between particle(s) and medium. In such systems there is potential for one particle to indirectly effect the trajectory of another through the medium, in which case the Lorentz assumption of noninteraction does not hold and it is natural to consider multiple particles. The challenge is that with multiple particles the complexity of the system is greatly increased. This is likely the major reason few rigorous results exist in this setting although numerous results are known for the single particle variants (see previous references).

In this paper, we study the motion of $N\geq 2$ particles on the regular two-dimensional triangluar lattice, in which the lattice is fully occupied by flipping scatterers. The particular type of scatterers we consider here are flipping rotators, which rotate the particle's velocity either to its left or its right by an angle of $\theta=\pm 2\pi/3$, depending on whether the scatterer is \emph{oriented} to the left or the right, i.e. is a left or right scatterer, respectively. Furthermore, the scatterers \emph{flip} or change orientation after scattering a particle, flipping either from right to left or from left to right, depending on their original orientation, respectively.

In most LLGs the system's initial configuration of scatterers has a significant impact on the particle's dynamics (see, for instance, \cite{Webb14,Bunimovich04}). In the case we consider here the model's initial configuration is much less relevant, which is our primary reason for using this model as a stepping stone for rigorously analyzing multiparticle systems in interacting media (see Theorem \ref{thm:1} in this paper and \cite{Grosfils99} for more details).

Here we show that, as opposed to the nearly linear motion of a single particle, the trajectories of multiple particles in this model can become entangled both over short and arbitrarily long time scales. The entanglement can result in both periodic and aperiodic trajectories that are not possible in the single particle model. Moreover, these types of entangled trajectories occur when the particles have the same speed and in specific cases when the particles have different speeds.

In the case of periodic orbits we show that not only do different types of periodic orbits exist, depending on the initial positions and velocities of the particles, but that these periodic structures can have arbitrarily large size and period (see Theorem \ref{thm:2}). We also show that, conversely, if some particle in the model has an unbounded, therefore aperiodic, past then some, potentially other particle, must have an unbounded future. That is, if a particle with an unbounded past becomes entangled with another particle causing the particle's motion to become periodic, some other particle must inherit this unbounded motion and escape to infinity (see Theorem \ref{thm:un} and Figure \ref{Fig:5}).

We also consider the case in which particles have different speeds. Although this complicates the analysis of our model we are able to show that periodic structures can still exists between particles with differing speeds (see Figure \ref{Fig:6}). However, we are only able to find relatively few such structures for a very limited number of particles with different speeds. However, we are able to give a sufficient condition involving the ratio of the particles speeds that guarantee that no periodic trajectories can form it this ratio is too high, i.e. if one particle is moving much faster than the other (see Theorem \ref{thm:5}).

The paper is organized as follows. In Section \ref{sec:2} we introduce the multiparticle model we will consider throughout the paper and describe the motion of a single particle in this model. In Section \ref{sec:3} we begin our analysis of the multiparticle system describing mutual and nonmutual interactions. We then discus the notion of particle entanglement and in the specific case of periodic trajectories note that periodic trajectories can be classified into two classes of regular and irregular trajectories (see Proposition \ref{prop:corner}). In Section \ref{sec:4} we consider the case of aperiodic behavior. Here we discuss unbounded behavior and describe consequences of the time-reversability of our model including the equivalence of periodic and bounded trajectories. In Section \ref{sec:5} we consider the effect of differing speeds on the particle's dynamics and show that some properties remain in this generalization of the multiparticle model while others do not. Section \ref{sec:6} consist of a number of open problems and some closing remarks.

We note that although the main results of this paper are proven rigorously, the paper is written so that it can be followed without the need for the reader to work through the various proofs.

\section{The Multiparticle Model}\label{sec:2}
In this section we describe the specific multiparticle model which we will consider throughout this paper. The lattice over which the particles move is the triangular lattice $T=(\mathbb{T},\mathbb{B})$, with \emph{sites} $\mathbb{T}$ and \emph{bonds} $\mathbb{B}$. This lattice consists of regular triangles with sides of unit length, so that each lattice site has six nearest neighbors with which it shares a lattice bond of length 1.

Our main object of study in this paper is the motion of a number of particles $p_1,\dots,p_N$ as they move from their initial positions on the lattice along the lattice bonds from lattice site to lattice site. By way of notation we let $\mathbf{r}_i(t)\in\mathbb{R}^2$ denote the \emph{position} and $\mathbf{v}_i(t)\in\mathbb{R}^2$ denote the \emph{velocity} of particle $p_i$ at time $t\geq 0$ for $i=1,\dots,N$. We let $\mathbf{r}_i=\mathbf{r}_i(0)$ and $\mathbf{v}_i=\mathbf{v}_i(0)$ denote the $i$th particle's \emph{initial position} and \emph{initial velocity}, respectively, at time $t=0$.

For the sake of simplicity we will initially assume that each particle moves with constant unit speed. Later in Section \ref{sec:5} we will remove this assumption to illustrate how varying speeds effect the collective dynamics of the particles in this model. Since, for the moment, each particle moves with unit speed, the particle $p_i$ is at some lattice site at each time $t+\triangle t_i$ where $t\in\mathbb{N}=\{0,1,2,\dots\}$ with \emph{displacement} $\triangle t_i\in [0,1)$. We define the particle $p_i$'s \emph{trajectory} as the discrete sequence of lattice sites $\{\mathbf{r}_i(t+\triangle t_i)\}_{t\in\mathbb{N}}\subseteq\mathbb{H}$, where we consider only the discrete times $t=0+\Delta t_i,1+\Delta t_i,2+\Delta t_i,\dots$ for each particle instead of all time $t\geq 0$. Since the velocity of the particle does not exist at the moment it is scattered, we let $\mathbf{v}_i(t+\triangle t_i)$ denote the velocity of the particle immediately \emph{after} it is scattered, i.e. immediately after it leaves a lattice site.

We also assume that the displacement $\triangle t_i\neq \triangle t_j$ for $i\neq j$ so that no two particle's arrive at the same lattice site at the same point in time. This is not a strong assumption in the sense that if we choose any probability measure on $[0,1)$ absolutely continuous with respect to Lebesgue measure then the probability of $\triangle t_i=\triangle t_j$ for $i\neq j$ is zero. Lastly, and without loss in generality, we assume that $\triangle t_1 =0$ and $\triangle t_i<\triangle t_{i+1}$.

Particles moving along the same bond in opposite directions do not interact when they meet. This is the ``Lorentz" property of the system. However, at each lattice site $\mathbf{h}\in\mathbb{H}$, we assume that there is a scatterer, which rotates the velocity of the incoming particle, either to its left or to its right, by an angle of $\theta=\pm 2\pi/3$, depending on the scatterer's orientation. This is shown in Figure \ref{Fig:1}, where we use the convention, here and throughout the paper, that a closed circle denotes a left rotator and an open circle denotes a right rotator, respectively. As mentioned in the introduction, when a particle collides with a scatterer it \emph{flips} the scatterer's orientation either from right to left or left to right, respectively.

\begin{figure}
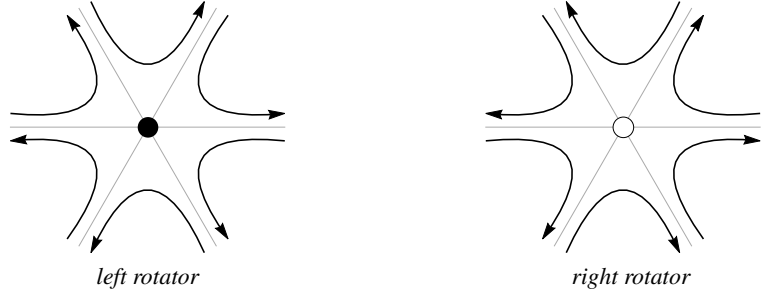

\begin{center}
\begin{tabular}{cc}
    \begin{overpic}[scale=.3]{TriLeftRotator.pdf}
    \put(32,-9){\small\emph{left rotator}}
    \end{overpic} & \hspace{2cm}
    \begin{overpic}[scale=.3]{TriRightRotator.pdf}
    \put(32,-9){\small\emph{right rotator}}
    \end{overpic}
\end{tabular}
    \vspace{0.2in}
\caption{Upon arriving at a left (right) rotator, indicated by a closed (open) circle, the particle's velocity is rotated to its left (right) by an angle of $\theta=\pm2\pi/3$. Immediately after the particle is scattered by the left (right) rotator the rotator's orientation flips so that it becomes a right (left) rotator.}\label{Fig:1}
\end{center}
\end{figure}

Note that each scatterer is, initially at time $t=0$, either a left scatterer or a right scatterer, i.e. is \emph{oriented} either to the left or to the right. With this in mind, we let $C=C(0)$ denote this \emph{initial configuration} of scatterers and let $C(t+\triangle t_i)$ denote the \emph{configuration} of scatterers on the lattice at time $t+\triangle t_i$ for $t\in\mathbb{N}$ and $i=1,\dots,N$. The configuration $C(t+\triangle t_i)$ consists of the collection of all the individual orientations of each scatter on the lattice at time $t+\triangle t_i$. For each lattice site $\mathbf{h}\in\mathbb{H}$ we let
\begin{equation*}
C(t+\triangle t_i,\mathbf{h})\in\{-1,1\} \ \ \text{for} \ \mathbf{h}\in\mathbb{H}, \ \ t\geq 0 \ \ \text{and} \ \ i=1,\dots,N
\end{equation*}
denote the \emph{orientation} of each scatterer at site $\mathbf{h}$ at time $t+\triangle t_i$. The orientation $C(t+\triangle t_i,\mathbf{h})=-1$ indicates, that at time $t+\triangle t_i$ the scatterer at lattice site $\mathbf{h}$ is a left scatterer, whereas the orientation $C(t+\triangle t_i,\mathbf{h})=1$ indicates, that the scatterer is a right scatterer. Furthermore, we let $C(\mathbf{h})\equiv C(0,\mathbf{h})$ denote the \emph{initial orientation} of the scatterer at $\mathbf{h}\in\mathbb{H}$ at time $t=0$.

Suppose each particle $p_i$ has initial position $\mathbf{r}_i$ and initial velocity $\mathbf{v}_i$. Then for an initial configuration $C$, we call $I=(\bar{\mathbf{r}},\bar{\mathbf{v}},C)$ an \emph{initial condition} where $\bar{\mathbf{r}}=(\mathbf{r}_1,\dots,\mathbf{r}_N)$ and $\bar{\mathbf{v}}=(\mathbf{v}_1,\dots,\mathbf{v}_N)$ are the collection of initial positions and velocities, respectively, of these particles. For an initial condition $I$, the $i$th particle's \emph{deterministic} equations of motion are given by
\begin{align}
\mathbf{r}_i(t+\triangle t_{j+1})&=
\begin{cases}
\mathbf{r}_i(t+\triangle t_i)+(\triangle t_{i+1}-\triangle t_{i})R[C(t+\triangle t_{i},\mathbf{r}_i(t+\triangle t_{i}))]\mathbf{v}_i(t+\triangle t_{i}) & \ \ \text{if} \ \ j=i\\
\mathbf{r}_i(t+\triangle t_{j})+(\triangle t_{j+1}-\triangle t_{j})\mathbf{v}_i(t+\triangle t_{j}) & \ \ \text{otherwise}\label{eq:1}
\end{cases}\\
\mathbf{v}_i(t+\triangle t_{j+1})&=
\begin{cases}
R[C(t+\triangle t_{i},\mathbf{r}_i(t+\triangle t_{i}))]\mathbf{v}_i(t+\triangle t_{i}) & \ \ \text{if} \ \ j=i\\
\mathbf{v}_i(t+\triangle t_{j}) & \ \ \text{otherwise}\label{eq:2}
\end{cases}\\
C(t+\triangle t_{j+1},\mathbf{h})&=
\begin{cases}
-C(t+\triangle t_{j},\mathbf{h}) &  \ \ \text{if} \ \  \mathbf{h}=\mathbf{r}_j(t+\triangle t_{j})\\
\hspace{0.1in} C(t+\triangle t_{j},\mathbf{h}) & \ \ \text{otherwise}\label{eq:3}
\end{cases}
\end{align}
for $t\in\mathbb{N}$, $1\leq i\leq N$, and $0\leq j\leq N$ where $\triangle t_{N+1}=1$. Equation \eqref{eq:1} gives the dynamics of the $i$th particle describing its piecewise linear motion between successive scatterings. The \emph{rotation operator} $R:\{-1,1\}\rightarrow\mathbb{R}^{2\times 2}$ in equation (\ref{eq:2}) is the rotation matrix
\begin{equation}\label{eq:matrix}
R[z]=
\left[\begin{array}{rr}
\cos\left(\frac{2\pi}{3}z\right)&\sin\left(\frac{2\pi}{3}z\right)\\
-\sin\left(\frac{2\pi}{3}z\right)&\cos\left(\frac{2\pi}{3}z\right)
\end{array}\right] \ \text{where} \ z\in\{-1,1\},
\end{equation}
which describes how the velocity of the particles are rotated, when a particle arrives at a scatterer. Equation \eqref{eq:3} describes the flipping motion of the scatterers.

Given an initial condition $I$, each of the particle's motion over the lattice is uniquely determined for all $t\geq0$ by equations \eqref{eq:1}--\eqref{eq:3}. This leads us to the following definition, which describes the general type of LLG we consider in this paper.

\begin{definition}\textbf{(The Multiparticle Model)}
Let $(T,I,N)$ denote the LLG with $N\geq1$ particles $p_1,p_2,\dots,p_N$, initial condition $I=(\bar{\mathbf{r}},\bar{\mathbf{v}},C)$, and equations of motion given by \eqref{eq:1}--\eqref{eq:3}. We will call this the \emph{multiparticle model} on the triangular lattice with $N$ particles and  initial condition $I$.
\end{definition}

\begin{figure}
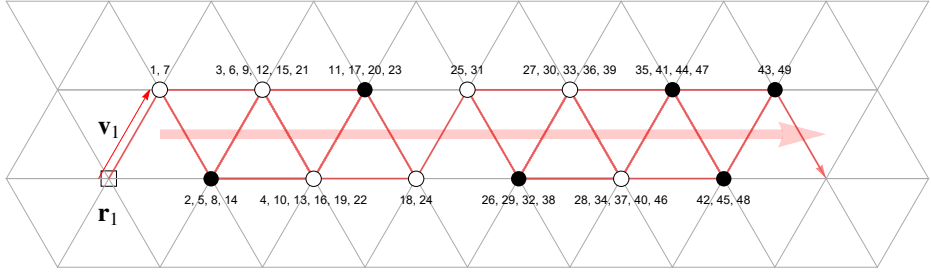

\begin{center}
\begin{overpic}[scale=.42]{modifiedSingle.pdf}
\put(11.5,16){$\mathbf{v}_1$}
\put(11.5,7){$\mathbf{r}_1$}
\end{overpic}\label{Fig:2}
\end{center}
\caption{The path of a single particle $p_1$ with initial position $\mathbf{r}_1=(0,0)$ indicated by the square, initial velocity $\mathbf{v}_1=(1/2,\sqrt{3}/2)$, and initial displacement $\Delta t_1=0$ is shown in which the particle moves through a horizontal strip for a randomly generated configuration of scatterers. The scatterer's initial orientation on the strip are shown. The numbers at each site in the strip indicate the time $t=1,2,\dots, 49$ the particle arrives at each site, respectively.}\label{Fig:2}
\end{figure}

In the case that $N=1$ we have a single particle and the following is known to hold.

\begin{theorem}\textbf{(Propagation of a Single Particle for any Initial Condition) \cite{Grosfils99}}\label{thm:1}
For any initial configuration of scatterers on a triangular lattice, the single particle in the $(T,I,1)$ model propagates in one particular direction through a strip on the lattice.
\end{theorem}

That is, the trajectory of a single particle is confined to a \emph{strip} in the triangular lattice, which is a region of the lattice bounded by two adjacent parallel lines. The motion of the particle in the strip is in a single direction in that any motion back towards its initial position is quickly blocked and the particle is forced forward (see Figure \ref{Fig:2}). A large part of \cite{Grosfils99} is devoted to describing the details of this blocking mechanism. The single particle \emph{propagates} in that its \emph{square-displacement} $\Delta(t)=||\mathbf{r}_1(0)-\mathbf{r}(t)||^2$ has the property that $\Delta(t)\sim t^2$, which up to a constant is the fastest a particle moving with unit speed can move through the lattice (see \cite{Webb14} for details).

It is worth noting that if the single particle in the $(T,I,1)$ model is observed from a distance its motion appears to be nearly linear as the particle moves in essentially a single direction in a strip that is two lattice sites wide. In this sense the particle can be thought of as having a motion that is approximately the motion of a particle moving through empty space under Newtonian laws of motion. However, the particle is in fact colliding with scatterers at every time step $t$ for all $t\in\mathbb{N}$.

In the following section we consider how this \emph{nearly linear motion} of a single particle can be effected by the presence of other particles as they move through the lattice.

\section{Multiple Particles, Entanglement, and Periodic Behavior}\label{sec:3}

The behavior of a single particle is important in the multiparticle model in that if a particle does not interact with another particle beyond some point in time then its motion from then on will be the nearly linear motion described in Theorem \ref{thm:1}. In fact, it is worth reiterating that particles in this model are assumed not to interact with each other. They do, however, indirectly \emph{interact} in that as one particle moves through the lattice it modifies the scatterers it collides with effecting the trajectory of other particles that later collide with these same scatterers.

\begin{figure}
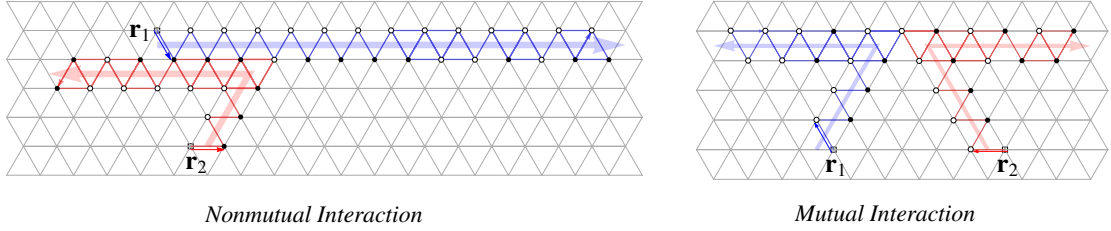

\begin{center}
\begin{tabular}{cc}
    \begin{overpic}[scale=.27]{1waycollision.pdf}
    \put(32,-5){\small\emph{Nonmutual Interaction}}
    \put(20.5,23){$\mathbf{r}_1$}
    \put(29,3){$\mathbf{r}_2$}
    \end{overpic} &
    \begin{overpic}[scale=.27]{modifiedDouble.pdf}
    \put(25,-7.5){\small\emph{Mutual Interaction}}
    \put(32,4){$\mathbf{r}_1$}
    \put(72,4){$\mathbf{r}_2$}
    \end{overpic}
\end{tabular}
    \vspace{0.1in}
\caption{Left: Particle $p_2$ interacts with particle $p_1$, shown in red and blue respectively, resulting in a deflection of its linear motion to the left. The interaction is nonmutual as particle $p_1$ does not interact with $p_2$ but continues its motion to the right. Right: A mutual interaction between $p_1$ and $p_2$ is shown, which causes the motion of $p_1$ to deflect to the left and the motion of $p_2$ to deflect to the right.}\label{Fig:Interactions}
\end{center}
\end{figure}

This notion of an \emph{indirect interaction} is defined as follows.

\begin{definition}\textbf{(Particle Interactions)}
Suppose in the $(T,I,N)$ model that particle $p_i$ collides with the scatterer at lattice site $\mathbf{h}\in\mathbb{H}$ at times $t_0$ and $t_1$, where $t_0<t_1$. Then $p_i$ \emph{interacts} with particle $p_j$ for $i\neq j$ at time $t_1$ if between time $t_0$ and $t_1$ particle $p_j$ collides with the scatterer at $\mathbf{h}\in\mathbb{H}$ an odd number of times. More generally, $p_i$ \emph{interacts} with a number of other particles if between consecutive collisions of $p_i$ with a particular scatterer these particles collectively collide with the same scatterer an odd number of times.
\end{definition}

Roughly speaking, particle $p_i$ interacts with another particle or some number of particles if, when $p_i$ returns to a previously visited scatterer, it finds the orientation of the scatterer has been flipped. It is worth noting that this notion of an \emph{interaction} is potentially \emph{nonmutual} in that $p_i$ may interact with $p_j$ but $p_j$ may not interact with $p_i$ (see Figure \ref{Fig:Interactions}).

One of the main differences between a single particle's trajectory and the trajectory of multiple particles is that two, or possibly more particles, can get entangled. By \emph{entangled} we mean the following.

\begin{definition}\textbf{(Multiparticle Entanglement)}
A subset $S$ of the particles in the $(T,I,N)$ model become \emph{entangled} if each particle in this subset interacts an infinite number of distinct times with at least one other particle in $S$ and this is not true for any proper subset of $S$.
\end{definition}

In the simplest case, two particles are said to \emph{entangle} if at least one of the particles continues to influence the trajectory of the other. An important special case of entanglement is a periodic orbit. The trajectory of a particle $p_i$ is \emph{periodic} if there is a smallest time $\tau_i<\infty$, such that the particle's position satisfies $\mathbf{r}_i(t+\triangle_i+\tau_i)=\mathbf{r}_i(t+\triangle_i)$ for each $t\in\mathbb{N}$ or equivalently $\mathbf{r}_i(t)=\mathbf{r}_i(t+\tau_i)$ for all $t\geq 0$. If this is the case we call $\tau_i$ the particle's \emph{period}. The \emph{size} of a periodic orbit is the number of distinct lattice sites the particle visits in one period of its motion.

\begin{figure}
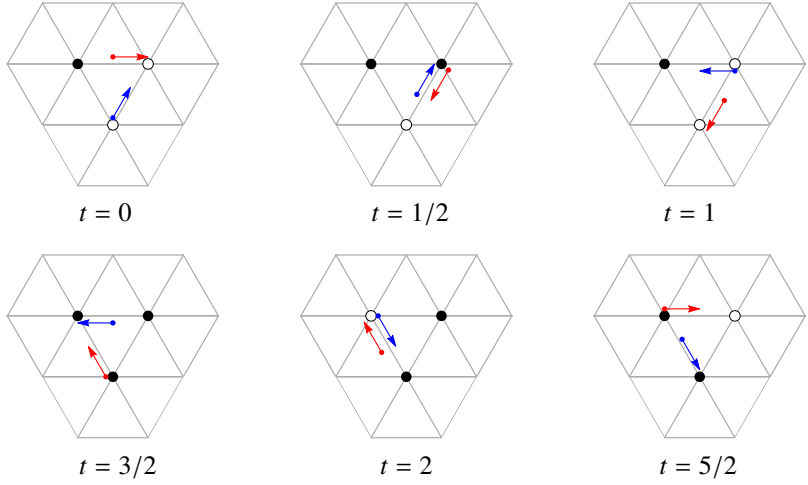

\begin{center}
\begin{tabular}{c}
	\begin{overpic}[scale=.9]{triangle_2by3.pdf}
    \put(12,30.5){$t=0$}
    \put(46,30.5){$t=1/2$}
     \put(80,30.5){$t=1$}
    \put(12,1){$t=3/2$}
     \put(47,1){$t=2$}
    \put(79.5,1){$t=5/2$}
    \end{overpic}
\end{tabular}
\end{center}
\caption{The simplest period orbit in the $(T,I,N)$ model is shown consisting of two particles moving along the vertices of a single triangle in opposite directions. The positions of the particles $p_1$ and $p_2$ are shown in blue and red, respectively, at times $t=0,1/2,1,\dots,5/2$; with blue and red arrows indicating their respective velocities. The displacement of these particles are $\Delta t_1=0$ and $\Delta t_2=1/2$. At each time $t$ the orientation of each scatterer on the triangle is shown.}\label{Fig:3}
\end{figure}

The reason a particle with a periodic trajectory is \emph{entangled} is that if there is a last time the particle interacts with another particle, then its trajectory is thereafter as described in Theorem \ref{thm:1}, which is nonperiodic. It is worth noting that even if a particle in the $(T,I,N)$ model has a periodic trajectory the behavior of the LLG as a system may not be periodic. The reason is that the multiparticle model $(T,I,N)$ is referred to as \emph{periodic} if there is a $\tau>0$ such that each of
\begin{equation}\label{eq:forward}
\bar{\mathbf{r}}(t+\tau)=\bar{\mathbf{r}}(t), \ \bar{\mathbf{v}}(t+\tau)=\bar{\mathbf{v}}(t), \ \text{and} \ C(t+\tau)=C(t) \ \text{for all} \ t\geq0.
\end{equation}
If $(T,I,N)$ is periodic then each particle's trajectory is also periodic. However, if only a fraction of the particles in the model have a periodic trajectory then the model as whole is not periodic.

The simplest form of periodicity is considered on the following example.

\begin{example}
Consider the two particle system $(T,I,2)$ with particles $p_1$ and $p_2$ shown in Figure \ref{Fig:3}. The positions of $p_1$ and $p_2$ are indicated by the blue and red dot for $t=0,1/2,1,\dots,5/2$, respectively. The corresponding blue and red arrows indicated the particle's velocity respectively. The particle's displacements are $\Delta t_1=0$ and $\Delta t_2=1/2$ so that exactly one particle is at a site at times $t=k/2$ for $k=0,1,2,\dots$. The two particles move periodically either clockwise or counterclockwise on a single triangle forming the smallest possible periodic orbit both in terms of size and period, both of which are 3.
\end{example}

One of the main results in this paper is that periodic orbits can be arbitrarily large in both size and period. This is a consequence of the existence of highly regular periodic orbits which can be extended indefinitely.

\begin{theorem}\textbf{(Existence of Arbitrarily Large Periodic Trajectories in the Multiparticle Model)}\label{thm:2}
For any $N\geq2$ there exist periodic trajectories in the multiparticle model $(T,I,N)$. Moreover, these trajectories can have arbitrarily large size and period.
\end{theorem}

\begin{proof}
Consider first the case in which $N=2$. Note that the two particles $p_1$ and $p_2$ in Figure \ref{Fig:4} have periodic trajectories of size $s(d)=3d+1$ and period $\tau(d)=4+4(d-1)$, where $d=||\mathbf{h}_1-\mathbf{h}_2||$ is the distance between the two lattice sites $\mathbf{h}_1,\mathbf{h}_2\in\mathbb{T}$. As both $\tau(d)$ and $s(d)$ are unbounded functions of the length $d$ the $(T,I,N)$ model can have periodic orbits of arbitrarily large period and size for $N=2$.

For $N>2$ suppose $p_1$ and $p_2$ have the periodic trajectories shown in Figure \ref{Fig:4}. We let the other $N-2$ particles have noninteracting trajectories that run in strips parallel to this periodic orbit (cf. Figure \ref{Fig:2}). This can be done by giving each particle the proper initial position, velocity, and configuration of scatterers in a neighborhood of its initial position as described in \cite{Grosfils99}. In this way the particles $p_1$ and $p_2$ do not interact with $p_3,\dots,p_N$. This proves the result.
\end{proof}

\begin{figure}
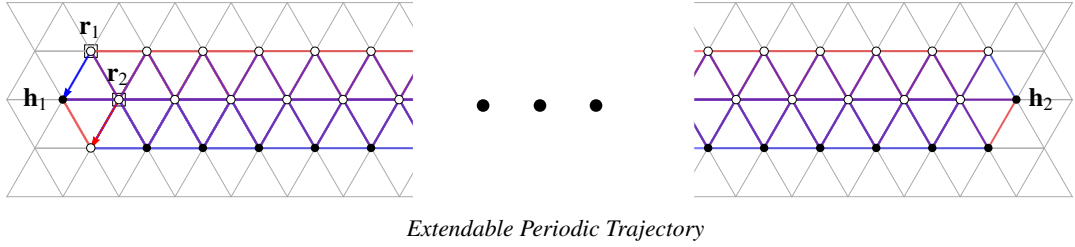

\begin{center}
	\begin{overpic}[scale=.32]{Extendable.pdf}
    \put(38,-1){\small\emph{Extendable Periodic Trajectory}}
    \put(8.5,17.25){$\mathbf{r}_1$}
    \put(11,13){$\mathbf{r}_2$}
    \put(3.5,10.75){$\mathbf{h}_1$}
    \put(94,10.75){$\mathbf{h}_2$}
    \end{overpic}
\end{center}
\caption{The periodic trajectories of two particles $p_1$ and $p_2$ are shown in blue and red, respectively. Here $\mathbf{r}_1$ and $\mathbf{r}_2$ indicate initial positions and arrows indicate initial velocities. The displacements $\Delta t_1$ and $\Delta t_2$ can be any nonequal numbers in $[0,1)$ but are shown for simplicity as $\Delta t_1=\Delta t_2=0$. The periodic obits of each point have size $s(d)=3d+1$ and period $\tau(d)=4+4(d-1)$ where $d=||\mathbf{h}_1-\mathbf{h}_2||$.}\label{Fig:4}
\end{figure}

Not all periodic trajectories have the form of the periodic trajectories shown in Figure \ref{Fig:4}. In Figure \ref{Fig:zoo} a number of much more \emph{irregular} periodic trajectories are shown. To distinguish between regular and irregular periodic trajectories we divide all possible particle velocities into two set
\[
V_1=\{(1,0),(-1/2,\sqrt{3}/2),(-1/2,-\sqrt{3}/2)\} \ \ \text{and} \ \ V_2=\{(-1,0),(1/2,-\sqrt{3}/2),(1/2,\sqrt{3}/2)\}.
\]
If a two-particle periodic orbit has particles whose velocities have \emph{matching parity}, i.e. the particles' initial velocities belong to the same set, we say that the periodic orbit is \emph{regular}. Otherwise, it is \emph{irregular}. Note that the periodic orbit in Figure \ref{Fig:4} is regular, whereas the periodic orbits in Figure \ref{Fig:zoo} are irregular.

One can show that regular orbits have a smoother boundary than irregular orbits. Specifically, irregular orbits can have \emph{corners}, which are triangles that share a side with only one other triangle in the orbit. No corners are possible in regular orbits (see Figure \ref{Fig:4} and Figure \ref{Fig:zoo}). This is formally stated in the following proposition.

\begin{prop}\label{prop:corner}
Suppose two particles form a periodic orbit in the $(T,I,2)$ model. If the particles have matching parity, the periodic orbit cannot have corners. If the two particles do not have matching parity then the orbit can have corners.
\end{prop}

\begin{proof}
First, note that if a particle moving at unit speed has velocity $\mathbf{v}(t)\in V_i$ then $\mathbf{v}(t+1)\in V_i$ for $i=1,2$. That is, any particle in the $(T,I,N)$ model for any $N\geq 1$ has velocities that are either in $V_1$ or $V_2$ but not both for all time. Hence, the directions a particle can move are fixed by the particle's initial velocity.

Suppose then that two particles $p_1$ and $p_2$ belong to the same periodic orbit and that these particles have matching parity. By way of contradiction, suppose that the periodic orbit has a corner. Note that a particle can only arrive at this corner $\mathbf{c}$ by traveling first along some bond $b_1$ then leaving along some other bond $b_2$ or first along $b_2$ then along $b_1$.

If particle $p_1$ arrives at $\mathbf{c}$ by moving along $b_1$ then it cannot ever arrive at $\mathbf{c}$ by moving along $b_2$. The reason is that moving along $b_1$ and $b_2$ towards $\mathbf{c}$ are velocities that belong to different sets $V_1$ and $V_2$ and particle $p_1$ can only have velocities in one. Since $p_2$ is assumed to have the same parity then $p_2$ can only ever arrive at $\mathbf{c}$ by moving along $b_1$.

Since the rotator at $\mathbf{c}$ must flip an even number of times during one period of the periodic orbit and must be visited at least once, there must be a second time a particle arrives at $\mathbf{c}$. Both the first and second time the particles must travel first down $b_1$ and leave along $b_2$. The second time, however, the rotator at $\mathbf{c}$ has flipped its orientation so this second particle cannot leave along $b_2$, a contradiction. Since there is a similar contradiction if $p_1$ arrives at $\mathbf{c}$ by moving along $b_2$, the periodic orbit cannot have a corner if the particles have the same parity.

Conversely, the periodic orbit can have corners if the particles do not have matching parity as can be seen in Figure \ref{Fig:zoo}.
\end{proof}

It is currently unknown whether any irregular periodic orbits can be decomposed and extended ad infinitum similar to the regular periodic trajectories in Figure \ref{Fig:4}. Moreover, although we have done extensive numerical testing, no tangles involving more than two particles have been found. It is therefore unknown whether the multiparticle model $(T,I,N)$ can have periodic dynamics if $N$ is odd. If $N$ is even then it is possible, for instance, to create $N/2$ nonoverlapping copies of the orbit shown in Figure \ref{Fig:4}, which results in periodic dynamics of the entire system.

We note that by slightly weakening the definition of periodicity it is possible to refer to parts of aperiodic systems as periodic. To make this precise we write $\Omega\subset T$ if $\Omega$ is a subset of the triangle lattice consisting of some subset of lattice sites and the bonds between them. If $\Omega$ has only a finite number of lattice sites we say $\Omega$ is \emph{finite} or \emph{bounded}.

\begin{definition}\textbf{(Local Periodicitiy)}
Let $\Omega\subset T$. If the restriction $(\Omega,I,N)$ of the multiparticle model $(T,I,N)$ to $\Omega$ is periodic then we say $(T,I,N)$ is \emph{locally periodic} on $\Omega$.
\end{definition}

\begin{figure}
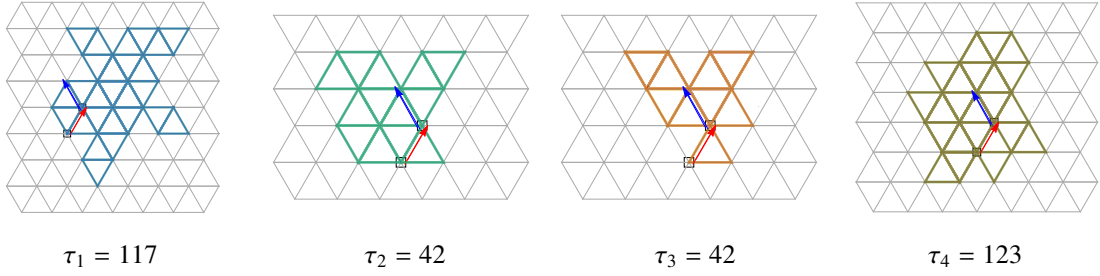

\begin{center}
\begin{tabular}{c}
	\begin{overpic}[scale=.45]{PeriodicZoo.pdf}
    \put(6,-3){$\tau_1=117$}
    \put(33,-3){$\tau_2=42$}
    \put(59,-3){$\tau_3=42$}
    \put(83.5,-3){$\tau_4=123$}
    \end{overpic}
\end{tabular}
\end{center}
\caption{Four irregular periodic orbits consisting of two particles each are shown with periods $\tau_i$ for $i=1,2,3,4$. In each figure the initial position and velocity of the particle's are indicated by a red and blue arrow, respectively.}\label{Fig:zoo}
\end{figure}

It is relatively straightforward to create a multiparticle model $(T,I,N)$ which is locally periodic for some $\Omega\subset T$ but overall aperiodic. For instance, combining the trajectories in Figures \ref{Fig:2} and \ref{Fig:4} in a nonoverlapping way results in such a system where if we let $\Omega$ contain the two particle periodic orbit we have local periodicity. However, if we slightly modify our multiparticle model by allowing an infinite number of particles it is possible to have a system in which each trajectory is periodic but the overall system is aperiodic.

\begin{example}
Consider the infinite multiparticle model $(T,I,\infty)$ in which we have $N=\infty$ particles. Suppose each particle is part of a pair of entangled periodic particles of the form shown in Figure \ref{Fig:4}, none of these pairs interact with any other particles, and there is no largest periodic orbit. Then each trajectory is periodic but the system taken as a whole is aperiodic since there is no $\tau<\infty$ such that $\bar{\mathbf{r}}(t)=\bar{\mathbf{r}}(t+\tau)$ for $t\geq 0$. In fact, although this model is aperiodic, if $\Omega$ is any subset of $T$ that is the union of a finite number of these periodic orbits then $(T,I,N)$ is locally periodic on $\Omega$.
\end{example}

It is also worth mentioning that not all entangled orbits are necessarily periodic. For instance, the two particle periodic orbit in Figure \ref{Fig:4} can be modified such that the particles $p_1$ and $p_2$ move to the right indefinitely. Although neither of the particles have a periodic trajectory in this situation the two particles are entangled for all time. In fact, much more complicated entanglements than this can be found in the $(T,I,N)$ model. (We save the analysis of these more complex entanglements for a latter paper.)

In the following section we more closely investigate what it means for a system to be aperiodic and what consequences this has on the past and future of the system.

\section{Aperiodic Behavior in the Multiparticle Model}\label{sec:4}

To understand what happens if a particle's trajectory is not periodic we first note that the dynamics of the particles $p_1,\dots,p_N$ are \emph{invertible}, i.e. the system's equations of motion are time-reversible. We can use this time-reversibility to show that the particles in the $(T,I,N)$ model each have a periodic trajectory if and only if they stay in a finite subset $\Omega$ of the (infinite) lattice $H$ for all time.

To show that each particle's equations of motion  \eqref{eq:1}--\eqref{eq:3} can be reversed, we note that these equations describe the $i$th particle's motion in \emph{forward time}. That is, given $\mathbf{r}_i(t+\triangle t_j)$, $\mathbf{v}_i(t+\triangle t_j)$, and $C(t+\triangle t_j)$ for $i=1,2,\dots,N$ we can compute each of these quantities at time $t+\triangle t_{j+1}$. In the following proposition, these quantities are shown to exist in \emph{reverse time}, i.e. given $\mathbf{r}_i(t+\triangle t_i)$, $\mathbf{v}_i(t+\triangle t_i)$, and $C(t+\triangle t_i)$ these quantities can be found at time $t+\triangle t_{j-1}$, (see equations \eqref{eq:1.1}--\eqref{eq:3.1}). The fact that these equations exist implies that, the particle's motion is time-reversible. Therefore, we can recover the state of the multiparticle model at any time $t<0$ if we know the system's initial state. This is summarized in the following proposition.

\begin{prop}\textbf{(Time-Reversability)}\label{prop:1}
For the initial condition $I=(\bar{\mathbf{r}},\bar{\mathbf{v}},C)$, the particle $p_i$ in the $(T,I,N)$ model has the time-reversed equations of motion given by
\begin{align}
\mathbf{r}_i(t+\triangle t_{j-1})&=\mathbf{r}_i(t+\triangle t_{j})-(\triangle t_{j}-\triangle t_{j-1})\mathbf{v}_i(t+\triangle t_{j})\label{eq:1.1}\\
\mathbf{v}_i(t+\triangle t_{j-1})&=
\begin{cases}
R^{-1}[C(t+\triangle t_{i},\mathbf{r}_i(t+\triangle t_{i}))]\mathbf{v}_i(t+\triangle t_{i+1}) & \ \ \text{if} \ \ j=i+1\\
\mathbf{v}_i(t+\triangle t_{j}) & \ \ \text{otherwise}\label{eq:2.1}
\end{cases}\\
C(t+\triangle t_{j-1},\mathbf{h})&=
\begin{cases}
-C(t+\triangle t_{j},\mathbf{h}) &  \ \ \text{if} \ \  \mathbf{h}=\mathbf{r}_{j-1}(t+\triangle t_{j-1})\\
\hspace{0.1in} C(t+\triangle t_{j},\mathbf{h}) & \ \ \text{otherwise}\label{eq:3.1}
\end{cases}
\end{align}
for $t\in\mathbb{Z}$, $1\leq i\leq N$, and $0\leq j\leq N$ where $\triangle t_{-1}=-1+\triangle t_N$. Here
\begin{align*}
\mathbf{r}_i(t+\triangle t_i)&=\mathbf{r}_i(t+\triangle t_{i+1})-(\triangle t_{i+1}-\triangle t_{i})\mathbf{v}_i(t+\triangle t_{i+1})\\
\mathbf{r}_{j-1}(t+\triangle t_{j-1})&=\mathbf{r}_{j-1}(t+\triangle t_j)-(\triangle t_j-\triangle t_{j-1})\mathbf{v}_i(t+\triangle t_j)
\end{align*}
in \eqref{eq:2.1} and \eqref{eq:3.1}, respectively, so that each quantity at time $t+\triangle t_{k-1}$ is given in terms of quantities from one step in the future at time $t+\triangle t_{k}$.
\end{prop}

The proof of this proposition is based on the observation that in the $(T,I,N)$ model, if the $i$th particle's velocity $\mathbf{v}(t+\triangle t_j)$ is known at time $t+\triangle t_j$, then there are only three possibilities for what the particle's velocity could have been at the previous time $t+\triangle t_{j-1}$. If the particle does not encounter a scatter at time $t+\triangle t_{j-1}$, i.e. $i\neq j-1$, then its velocity is unchanged. If it encounters a right rotator at time at time $t+\triangle t_{j-1}$, then the particle's velocity $\textbf{v}(t+\triangle t_{j-1})$ will be one of the the other two possibilities. If the particle encounters a left rotator at time $t+\triangle t_{j-1}$ then $\mathbf{v}(t+\triangle t_{j-1})$ will be the third possibility. Since it is possible to uniquely recover $\mathbf{v}(t+\triangle t_{j-1})$, based on the type of scatterer the particle encounters at this time, it is possible to uniquely determine the particles position $\mathbf{r}_i(t+\triangle t_{j-1})$. Therefore, it is possible not only to know the particle's future trajectory but also its past. A proof of proposition \ref{prop:1} is the following.

\begin{proof}
To prove the proposition note that by substituting $j+1$ for $j$ in Equations \eqref{eq:1.1}--\eqref{eq:3.1} results in the equations of motion
\begin{align}
\mathbf{r}_i(t+\triangle t_{j})&=\mathbf{r}_i(t+\triangle t_{j+1})-(\triangle t_{j+1}-\triangle t_{j})\mathbf{v}_i(t+\triangle t_{j+1})\label{eq:1.2}\\
\mathbf{v}_i(t+\triangle t_{j})&=
\begin{cases}
R^{-1}[C(t+\triangle t_{i},\mathbf{r}_i(t+\triangle t_{i}))]\mathbf{v}_i(t+\triangle t_{i+1}) & \ \ \text{if} \ \ j=i\\
\mathbf{v}_i(t+\triangle t_{j+1}) & \ \ \text{otherwise}\label{eq:2.2}
\end{cases}\\
C(t+\triangle t_{j},\mathbf{h})&=
\begin{cases}
-C(t+\triangle t_{j+1},\mathbf{h}) &  \ \ \text{if} \ \  \mathbf{h}=\mathbf{r}_{j}(t+\triangle t_{j})\\
\hspace{0.1in} C(t+\triangle t_{j+1},\mathbf{h}) & \ \ \text{otherwise}\label{eq:3.2}
\end{cases}
\end{align}
for $t\in\mathbb{Z}$, $1\leq i\leq N$, and $-1\leq j\leq N-1$. The goal is to show that evolving each of the particle's position, velocity, and the configuration of scatterers on the lattice first forward in time by Equations \eqref{eq:1}--\eqref{eq:3} then back in time by Equations \eqref{eq:1.2}--\eqref{eq:3.2} results in these quantities at the present time.

Going in this order, by inserting $\mathbf{r}_i(t+\triangle t_{j+1})$ and $\mathbf{v}_i(t+\triangle t_{j+1})$ from Equations \eqref{eq:1} and \eqref{eq:2}, respectively, into right-hand side of Equation \eqref{eq:1.2} for the case $j=i$ yields the equation
\begin{align*}
\mathbf{r}_i(t+\triangle t_i)=\mathbf{r}_i(t+\triangle t_{i})&+(\triangle t_{i+1}-\triangle t_{i})R[C(t+\triangle t_{i},\mathbf{r}_i(t+\triangle t_{i}))]\mathbf{v}_i(t+\triangle t_{i})\\
&-(\triangle t_{i+1}-\triangle t_{i})R[C(t+\triangle t_{i},\mathbf{r}_i(t+\triangle t_{i}))]\mathbf{v}_i(t+\triangle t_{i})\\
&=\mathbf{r}_i(t+\triangle t_i).
\end{align*}
Hence, by evolving the $i$th particle's position first forward in time using Equations \eqref{eq:1}--\eqref{eq:3} then backward in time using Equations \eqref{eq:1.1}--\eqref{eq:3.1} as given by Equations \eqref{eq:1.2}--\eqref{eq:3.2} we recover the particle's present state. For the case $j\neq i$ we similarly have
\begin{align*}
\mathbf{r}_i(t+\triangle t_j)=\mathbf{r}_i(t+\triangle t_{j})&+(\triangle t_{j+1}-\triangle t_{j})\mathbf{v}_i(t+\triangle t_{j})\\
&-(\triangle t_{j+1}-\triangle t_{j})\mathbf{v}_i(t+\triangle t_{j})\\
&=\mathbf{r}_i(t+\triangle t_j).
\end{align*}

To verify the same for the particle's velocity we insert $\mathbf{v}_i(t+\triangle t_{i+1})$ from Equation \eqref{eq:2} into the right-hand side of Equation \eqref{eq:2.2}. For the case $j=i$ this yields
\begin{align*}
\mathbf{v}_i(t+\triangle t_i)&=R^{-1}[C(t+\triangle t_{i},\mathbf{r}_i(t+\triangle t_{i}))]R[C(t+\triangle t_{i},\mathbf{r}_i(t+\triangle t_{i}))]\mathbf{v}_i(t+\triangle t_{i})\\
&=\mathbf{v}_i(t+\triangle t_i).
\end{align*}
For $j\neq i$, $\mathbf{v}_i(t+\triangle t_j)=\mathbf{v}_i(t+\triangle t_{j+1})=\mathbf{v}_i(t+\triangle t_j)$ where the first equality follows from Equation \eqref{eq:2.2} and the second from Equation \eqref{eq:2}. Hence, evolving the particle's velocity forward in time then back results in the particle's present velocity.

To verify that this property also holds for the model's configuration of scatterers we insert Equation \eqref{eq:3} into the right-hand side of \eqref{eq:3.2}. This yields
\begin{align*}
C(t+\triangle t_{j},\mathbf{h})
&=\begin{cases}
-C(t+\triangle t_{j+1},\mathbf{h}) &  \ \ \text{if} \ \  \mathbf{h}=\mathbf{r}_{j}(t+\triangle t_{j})\\
\hspace{0.1in} C(t+\triangle t_{j+1},\mathbf{h}) & \ \ \text{otherwise}
\end{cases}\\
&=\begin{cases}
-(-C(t+\triangle t_{j},\mathbf{h})) &  \ \ \text{if} \ \  \mathbf{h}=\mathbf{r}_{j}(t+\triangle t_{j})\\
\hspace{0.1in} C(t+\triangle t_{j},\mathbf{h}) & \ \ \text{otherwise}
\end{cases}=C(t+\triangle t_{j},\mathbf{h})
\end{align*}
where the first equality follows from Equation \eqref{eq:3.2} and the second from Equation \eqref{eq:3}. As before, evolving the scatterer's configuration at lattice site $\mathbf{h}$ first forward then backward in time recovers the scatterer's present orientation. This implies that the system's equations of motion \eqref{eq:1}--\eqref{eq:3} are time reversible and are given by Equations \eqref{eq:1.2}--\eqref{eq:3.2}.
\end{proof}

Suppose each particle in the $(T,I,N)$ model remains in a subset $\Omega\subset T$ for all $t\geq 0$. If $\Omega$ is finite we say that the trajectory of each particle is \emph{bounded} for $t\geq 0$. If this is the case there must be two integer-valued times $0\leq t_1<t_2$ at which $\bar{\mathbf{r}}(t_1)=\bar{\mathbf{r}}(t_2)$, $\bar{\mathbf{v}}(t_1)=\bar{\mathbf{v}}(t_2)$, and $C(t_1)=C(t_2)$. This is because of the discrete nature of the lattice. Each particle $p_i$ can only assume a finite number of positions along the sites and bonds of $\Omega$ as each particle moves with unit speed and we only consider times $t+\Delta t_i$ for $i=1,2,\dots,N$ and $t\in\mathbb{N}$. Similarly, there are only a finite number of velocities the particle can have and there are only a finite number of scattering configurations possible on $\Omega$. Therefore, at some first time $t_2$ each of the particle's position, velocities, and the configuration of scatterers on $\Omega$ must be the same as at some previous point in time $t_1<t_2$. Hence, the system's behavior must be periodic for all time $t\geq t_1$ with period $\tau=t_2-t_1>0$. However, this behavior may be only eventually periodic.

Formally, the multiparticle model's motion is said to be \emph{eventually periodic} with \emph{period} $\tau<\infty$, if there is a $t_*>0$, such that
\[
\bar{\mathbf{r}}(t)=\bar{\mathbf{r}}(t+\tau), \ \bar{\mathbf{v}}(t)=\bar{\mathbf{v}}(t+\tau), \ \text{and} \ C(t)=C(t+\tau) \ \text{for all} \ t\geq t_*.
\]
Importantly, if $t_*=0$, we do not consider the system's behavior to be eventually periodic, since it is then \emph{periodic} with period $\tau$.

However, since the dynamics of the $(T,I,N)$ model is time-reversible by Proposition \ref{prop:1}, then it dynamics cannot be eventually periodic only periodic (see for instance \cite{Brin2002}). This implies that if the trajectory of each particle in the $(T,I,N)$ model is bounded then the system's dynamics are periodic.

Being periodic has additional consequences as $(T,I,N)$ is time reversible. That is, if the $(T,I,N)$ model is periodic, i.e. Equation \eqref{eq:forward} holds for all $t\geq 0$, then this equation also holds for $t<0$ so that the system is \emph{periodic for all time}. The reason is that if at some time $t<0$ the systems dynamics are not periodic but become periodic at time $t=0$ then the system's dynamics are eventually periodic, which is not possible.

Since being periodic for all time implies that the trajectory of each particle is \emph{bounded for all time}, i.e. each particle stays in a finite set $\Omega$ for all time $t\in(-\infty,\infty)$, this implies the following result.

\begin{prop}\label{prop:1}\textbf{(Equivalence of Boundedness and Periodicity)}
The following are equivalent:\\
(i) The trajectory of each particle in the $(T,I,N)$ model is bounded for $t\geq 0$.\\
(ii) The trajectory of each particle in the $(T,I,N)$ model is bounded for all time $t\in(-\infty,\infty)$.\\
(iii) The $(T,I,N)$ model is periodic for $t\geq 0$.\\
(iv) The $(T,I,N)$ model is periodic for all time $t\in(-\infty,\infty)$.
\end{prop}

\begin{figure}
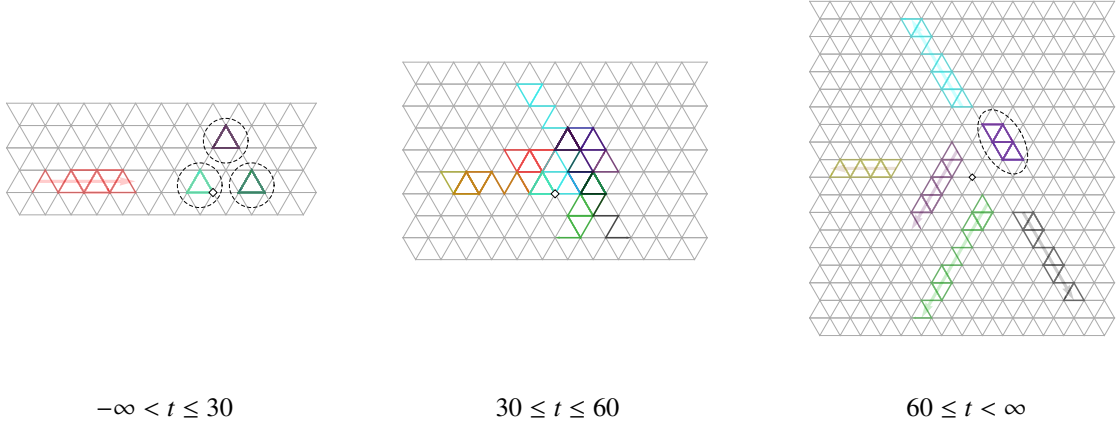

\begin{center}
\begin{overpic}[scale=.35]{MultiparticleGraphics1.pdf}
\put(9,-3){$-\infty < t \leq 30$}
\put(44,-3){$30 \leq t\leq 60$}
\put(80,-3){$60\leq t <\infty$}
\end{overpic}
\end{center}
\caption{Left: A particle, shown in red, with an unbounded past approaches three pairs of particles with periodic orbits indicated by dashed lines of the type shown in Figure \ref{Fig:3}. Center: The red particle collides with the light blue periodic orbit causing a chain reaction in which the trajectory of all particles is effected. Right: The yellow, purple, blue, green, and black particles escape off to infinity. The remaining two particles, which include the red particle with the unbounded past, become entangled in a periodic orbit indicated by the dashed line.}\label{Fig:5}
\end{figure}

It is unknown whether a single particle in the $(T,I,N)$ model with $N>1$ can have an aperiodic trajectory but also be bounded. The reason is that it may be possible for a particle to remain in a bounded region but sequentially interact with an aperiodic particle over longer and longer time intervals and thereby inherit the other particle's aperiodicity.

An important consequence of the previous proposition is that if a single particle has an unbounded trajectory in backward (forward) time then some other particle has an unbounded trajectory in forward (backward) time. The reason is that if some particle has an unbounded past, for instance, then it cannot have a bounded future as this would violate Proposition \ref{prop:1}. Hence, some possibly different particle must have an unbounded future.

This is summarized in the following theorem.

\begin{theorem}\label{thm:un}\textbf{(Unbounded Past and Future)}
If some particle in the $(T,I,N)$ model has an unbounded past (future) then some particle in the model has an unbounded future (past).
\end{theorem}

It is worth noting that a particle in the $(T,I,N)$ model can have both an aperiodic past and a periodic future or an aperiodic future and a periodic past. This is demonstrated in the following example that illustrates Theorem \ref{thm:un}.

\begin{example}
In Figure \ref{Fig:5} (left) a particle with an unbounded past, shown in red, collides with a number of particles with periodic trajectories. In the process the particle becomes entangled in a new periodic orbit with one of the other six particles. Since the red particle has an unbounded past then Theorem \ref{thm:un} implies that at least one of the other five particles must have an unbounded future. In fact, all other particles escape to infinity as can be seen in Figure \ref{Fig:5} (right).
\end{example}

It is also worth noting that Proposition \ref{prop:1} and Theorem \ref{thm:un} hold for any multiparticle lattice system in which the equations of motion are time reversible. That is, the fact that at least one particle in the system inherits an unbounded future from some particle's unbounded past is a feature shared by any such model. This is worth emphasizing as there are few rigorous results for multiparticle lattice systems with an arbitrary number of particles.

\section{Multiparticle Model for Particles with Different Speeds}\label{sec:5}
In the previous sections each of the particles in the $(T,I,N)$ model is considered to move at unit speed or more generally at the same speed. Here, we consider the case in which the particles can move at different but fixed speeds. That is, each particle $p_i$ is given the initial velocity $\mathbf{v}_i\in\mathbb{R}^2$ where $||\mathbf{v}_i(t)||=||\mathbf{v}_i||>0$ for all $t$ at which $\mathbf{v}_i(t)$ exists.



To understand to what degree we might expect particles with varying speeds to have periodic verses aperiodic behavior, we note the following. If we were to randomly choose the particles' speed in the $(T,I,N)$ model based on some probability measure on $[0,\infty)$ that is absolutely continuous with respect to Lebesgue measure then the probability of a periodic trajectory will be zero. The reason is the following result.

\begin{prop}\label{prop:3}
In the $(T,I,N)$ model with $N\geq2$ if $||\mathbf{v}_i||/||\mathbf{v}_j||$ is irrational for $i\neq j$ then particles $p_i$ and $p_j$ cannot be part of the same periodic orbit. Consequently, if $||\mathbf{v}_i||/||\mathbf{v}_j||$ is irrational for all $1\leq i,j,\leq N$ then each particle has an (unbounded) aperiodic trajectory.
\end{prop}

\begin{proof}
If particles $p_i$ and $p_j$ are part of the same periodic orbit of period $\tau>0$ then both $||\mathbf{v}_i||\cdot\tau$ and $||\mathbf{v}_j||\cdot\tau$ are whole numbers. The reason is $||\mathbf{v}_k||\cdot\tau$ is the length of the path that particle $p_k$ takes in one period of its orbit for $k=i,j$. Since the triangular lattice we are using has bond with unit length this distance must be a positive integer. Thus, $||\mathbf{v}_i||/||\mathbf{v}_j||=(||\mathbf{v}_i||\cdot\tau)/(||\mathbf{v}_j||\cdot\tau)$ must be a rational number. The result then follows.
\end{proof}

If a single particle does not interact at any point in time with any other particle then its trajectory, irrespective of the particle's speed, will be as described in Theorem \ref{thm:1}, i.e. the particle will propagate in a single direction in a strip of width 1. A natural question is whether two or more particles can become entangled when the particle's move at \emph{rationally related speeds}, i.e. when the ratio of their speeds $||\mathbf{v}_i||/||\mathbf{v}_j||$ is a rational number. As it turns out, this can happen as is demonstrated in the following example.

\begin{example}\label{ex:2part}
Consider the two particle system with particles $p_1$ and $p_2$ shown in Figure \ref{Fig:6}. Here $p_1$ has speed $||\mathbf{v}_1||=1$ and $p_2$ has speed $||\mathbf{v}_2||=2$. The trajectory of the two particles are entangled forming a periodic orbit with period $\tau=36$ and size $s=13$.
\end{example}

\begin{figure}
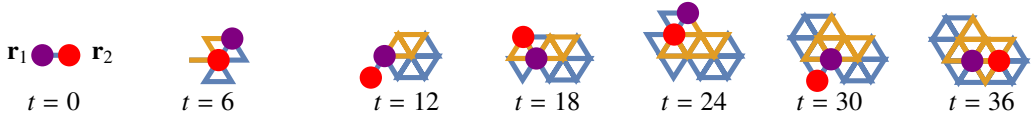

\begin{center}
\begin{tabular}{ccccccc}
	\begin{overpic}[scale=1.2]{multispeed1.pdf}
    \put(15,30){$\mathbf{r}_1$}
    \put(90,30){$\mathbf{r}_2$}
    \put(33,-15){$t=0$}
    \put(170,-15){$t=6$}
    \put(340,-15){$t=12$}
    \put(465,-15){$t=18$}
    \put(595,-15){$t=24$}
    \put(715,-15){$t=30$}
    \put(850,-15){$t=36$}
    \end{overpic} &
    \begin{overpic}[scale=1.2]{multispeed2.pdf}
    \end{overpic} &
    \begin{overpic}[scale=1.2]{multispeed3.pdf}
    \end{overpic} &
    \begin{overpic}[scale=1.2]{multispeed4.pdf}
    \end{overpic} &
    \begin{overpic}[scale=1.2]{multispeed5.pdf}
    \end{overpic} &
    \begin{overpic}[scale=1.2]{multispeed6.pdf}
    \end{overpic} &
    \begin{overpic}[scale=1.2]{multispeed7.pdf}
    \end{overpic}
\end{tabular}
\end{center}
\caption{A two particle periodic orbit in which particle $p_1$, shown in purple, has speed $||\mathbf{v}_1||=1$ and particle $p_2$, shown in red, has speed $||\mathbf{v}_2||=2$. The particles have period $\tau=36$.}\label{Fig:6}
\end{figure}

Although we have run extensive numerical tests for finding other such orbits, we have only found one other periodic trajectory for particles with different speeds. The other periodic orbit is similar to the one in Figure \ref{Fig:6} in that it consists of two particles one with speed 1 and the other with speed 2. Whether entanglements can happen for only two particles and whether they can happen only for particles moving the same speed (see Section \ref{sec:3}) or for one particle moving twice as fast as the other (see Example \ref{ex:2part}) remains an open question.

We note that in the case of $N=2$ particles, if one particle is moving much faster than the other it, is not possible for the two to form a periodic orbit.

\begin{theorem}\label{thm:5}
In the $(T,I,2)$ model if $||\mathbf{v}_1||/||\mathbf{v}_2||\geq 30$ then for any initial condition both particles escape to infinity.
\end{theorem}

\begin{proof}
The proof that a single particle always propagates in a strip of width one is based on the notion of blocking mechanisms. A \emph{blocking mechanism} of a particle is simply part of the particle's trajectory in which the particle is guaranteed to move one lattice site away from its initial position. During a blocking mechanism the particle cannot move back more than one lattice site towards its initial position. The original proof of Theorem \ref{thm:1} is based on the fact that a single particle's trajectory can be partitioned into disjoint blocking mechanisms (see \cite{Grosfils99} for more details). Each blocking mechanism has a duration of 2 to 7 time-steps for a particle moving at unit speed. Moreover, the first \emph{complete} blocking mechanism a particle experiences after time zero happens within 10 time steps.

Note that in the $(T,I,2)$ model if $p_2$ never interacts with $p_1$ then both particles must escape to infinity. For the case in which $p_2$ does interact with $p_1$ we suppose, without loss in generality, that $||\mathbf{v}_1||=1$ and $||\mathbf{v}_2||=1/30$. If $p_2$ interacts with $p_1$ at time $t_*$ then there is some time $t_1<t_*$ at which $\mathbf{r}(t_1)=\mathbf{r}(t_*)=\mathbf{h}$ where time $t_1$ was the last time $p_1$ visited lattice site $\mathbf{h}$ and $p_2$ visited $\mathbf{h}$ at some time between $t_1$ and $t_2$. By checking each blocking mechanism, the largest number of steps it takes for a particle traveling at unit speed to return to a previously visited lattice site is $6$. Hence, letting $t_1=0$ then $t_*\leq 6$. As $||\mathbf{v}_2||=1/30$ this means only at some time $t_2\geq 30-6=24$ can $p_2$ reach a lattice site adjacent to $\mathbf{h}$. In particular, between times $t_*=0$ and $t_2=24$ particle $p_2$ is on a single lattice bond adjacent to $\mathbf{h}$.

Note that by time $t_2=24$ that $p_1$ will have passed at least through its first, second, and third blocking mechanisms, where the first is guaranteed within 10 time steps and the other two in 7 steps a piece. Since each blocking mechanism moves the particle one lattice site away from its initial position and $p_2$ could only have arrived at a site distance one from this site at time $t_2=24$ then $||\mathbf{r}_1(t_2)-\mathbf{r}_2(t_2)||\geq 2$. Importantly, $p_1$ cannot revisit $\mathbf{r}_2(t_2)$ since it lies at least distance 2 in the opposite direction from the one it is moving in its strip.

The claim then is that beyond time $t_2=24$ particle $p_2$ cannot interact with particle $p_1$. The reason is that the slowest $p_1$ can move through each blocking mechanism is $||\mathbf{v}_1(bl)||\geq \frac{1}{7}\cdot1=\frac{1}{7}$ since the longest blocking mechanism moves the particle a distance 1 in 7 unit steps and the particle is moving at unit speed. Similarly, the fastest that $p_2$ can move through a blocking mechanism is $||\mathbf{v}_2(bl)||\leq \frac{1}{2}\cdot\frac{1}{30}=\frac{1}{60}$. Hence, in at most seven time steps $p_1$ will have moved another lattice site along its strip away from its initial position not encountering any lattice sites $p_2$ has visited since before time $t=0$, since it cannot visit $\mathbf{r}_2(t_1)$ or $\mathbf{r}_2(t_1)$. Continuing in this manner, as $p_1$ moves much faster through any blocking mechanism than $p_2$ then it follows that $p_2$ cannot interact with $p_1$ beyond time $t_*=6$. Hence, both particles escape to infinity.
\end{proof}

The situation is much more complicated if we want to determine whether periodic orbits can or cannot exist for particles with differing speeds if $N>2$. The reason is that even if one particle is moving much faster than the others, a second slower particle can still interact with the faster particle possibly sending it toward a third or fourth, etc. Ruling out the possibility that the fast particle is not somehow caught between a number of much slower particles is quite challenging and how to extend Proposition \ref{prop:3} to the case in which there are $N>2$ particles is an open question.





\section{Conclusion}\label{sec:6}
In this paper we consider how moving from the single particle model $(T,I,1)$ to a multiparticle model $(T,I,N)$ changes the dynamics of the system. In the original system a single particle has a ``nearly" linear motion, at least viewed macroscopically, so that the particle appears to move in one direction as if in a vacumm. However, the microscopic interaction with the system's media means that two noninteracting particles can become entangled in a periodic motion, etc. which is a behavior that would not be observed if the particles were moving through empty space.

This dramatic change in dynamics can also be observed in other LLG models. If the same system considered in this paper is put on the hexagonal lattice, our numerical experiments indicate there is a transition from periodic dynamics for a single particle (see \cite{Webb14}), to subdiffusive behavior for multiple particles. Similarly, on the square lattice a single particle will have an aperiodic and therefore unbounded trajectory, but our numerical experiments show, similar to what is observed in this paper, that multiple particles in this model can become entangled and form periodic trajectories. Currently, it is unknown whether rigorous results similar to those in this paper can be established for these other related multiparticle systems, although Proposition \ref{prop:1} and Theorem \ref{thm:un} can be directly extended to these related models.

For the $(T,I,N)$ model considered in this paper there are also a number of open questions. Although many different periodic orbits have been identified, it is still unknown whether periodic orbits involving more than two particles exists. Similarly, it is unknown whether periodic orbits involving particles with different speeds, other than those given in Section \ref{sec:5}, exist and whether orbits other than the one given in Figure \ref{Fig:4} can be infinitely extended.

It may also be possible to study this system in the  limiting case where $N=\infty$ using techniques from statistical mechanics. However, it is worth noting that this is not that same as studying this model in the case where $N<\infty$ and we place the triangular lattice on the torus, etc. If the lattice is finite then by our results the particles trajectories, no matter the number, will be periodic. This is likely not the case on the infinite triangular lattice, in which we fix some initial density of particles, since particles may travel infinitely far, especially if there is a low density of other particles.

\end{document}